\documentclass[11pt]{article}

\usepackage[margin=1.1in]{geometry}
\usepackage{amsmath,amssymb,amsthm}
\usepackage{xcolor}
\usepackage[numbers,sort&compress]{natbib}
\usepackage[colorlinks=true,linkcolor=blue!60!black,citecolor=blue!60!black,urlcolor=blue!60!black]{hyperref}

\newtheorem{theorem}{Theorem}
\newtheorem{corollary}{Corollary}
\newtheorem{lemma}{Lemma}

\theoremstyle{remark}
\newtheorem{remark}{Remark}

\newcommand{\cY}{\mathcal{Y}}
\newcommand{\E}{\mathbb{E}}
\newcommand{\KL}[2]{D(#1 \,\|\, #2)}
\newcommand{\JS}{\mathrm{JS}}
\newcommand{\Dbar}{\bar D_0}
\newcommand{\cD}{\mathcal{D}}

\title{The Optimal Discounting Parameter of the Power Prior\\
under Predictive Log-Loss}

\author{Yuriy A. Reznik\thanks{Massachusetts Institute of Technology,
Cambridge, MA. E-mail: \texttt{yreznik@mit.edu}.}}

\date{August 2026}

\begin{document}

\maketitle

\begin{abstract}
The power prior of Ibrahim and Chen incorporates historical data
into a Bayesian analysis by raising the historical likelihood to a
power $a_0 \in [0, 1]$. The choice of the exponent has remained an
open question. This paper gives a closed-form answer under the
predictive log-loss. For a model with $d$ parameters, a historical
sample of size $N_0$, and average Kullback--Leibler divergence
$\Dbar$ between the historical and current data-generating
distributions, the optimal exponent is
$a_0^{*} = d/(2 N_0 \Dbar + d)$. Equivalently, the optimally
borrowed effective sample size obeys the harmonic law
$1/E^{*} = 1/N_0 + 2\Dbar/d$: compatible data are pooled in full,
and any difference caps the borrowed information at $d/(2\Dbar)$
observations. The result is exact for multinomial data and extends
to smooth parametric families. The law benchmarks adaptive
borrowing, explains the reported degeneracy of the normalized power
prior, and shows that neither subsetting the data nor decaying the
exponent improves on the correctly discounted constant.
\end{abstract}

\section{Introduction}
\label{sec:intro}

The \emph{power prior} of Ibrahim and Chen~\cite{IbrahimChen00} is one of the
standard devices for incorporating historical data into a Bayesian
analysis. Given a statistical model with parameter~$\theta$,
historical data~$\cD_0$ with likelihood $L(\theta \mid \cD_0)$, and an
initial prior $\pi_0(\theta)$, it takes the form
\begin{equation}
\label{eq:pp}
\pi(\theta \mid \cD_0, a_0)
\;\propto\;
L(\theta \mid \cD_0)^{a_0}\, \pi_0(\theta),
\qquad
a_0 \in [0, 1],
\end{equation}
so that~$a_0 = 0$ discards the historical data and~$a_0 = 1$ pools
them with the current data at full weight. From the outset, the
\emph{discounting parameter} $a_0$ was conceived as a quantitative
expression of
heterogeneity: in the original formulation, $a_0$ is a scalar
quantifying the difference between the current and the historical
data~\cite{IbrahimChen00}. The construction has since been applied
across clinical trials, toxicology, environmental science, and survey
inference, and its variants form an active literature~\cite{IbrahimChen15}.

The choice of~$a_0$ has remained the construction's open question.
Ibrahim, Chen, and Sinha~\cite{IbrahimChenSinha03} settled the
question of \emph{form}. For fixed~$a_0$, the power-prior posterior
minimizes a convex combination of two Kullback--Leibler divergences,
anchored at the no-borrowing and full-borrowing posteriors, so the
geometry of partial borrowing is optimal. The \emph{amount} of
borrowing is another matter. Fixed values of~$a_0$ are typically
explored by sensitivity analysis. Treating~$a_0$ as random leads to
the \emph{normalized power prior} of Duan et al.~\cite{Duan06} and
Neuenschwander et al.~\cite{Neuenschwander09}. Its connection to
hierarchical models was developed by Chen and Ibrahim~\cite{ChenIbrahim06}, and it anchors the modern adaptive borrowing
methodology, alongside alternatives such as commensurate priors~\cite{Hobbs11}. The prior effective sample size, approximately~$a_0 N_0$ for a historical dataset of size~$N_0$, gives the parameter
its operational meaning~\cite{Morita08}.

Recent theory shows that the adaptive route resolves the selection
problem only partially. Pawel, Aust, and Held~\cite{Pawel23} proved
that normalized power priors always discount. Shen, Carvalho, Psioda,
and Ibrahim~\cite{Shen26} proved that in generalized linear models
the marginal posterior of~$a_0$ collapses to a point mass at zero
under any discrepancy between the historical and current data, while
failing to converge to one when the two are fully compatible. The
same authors construct optimal hyperpriors for~$a_0$ by minimizing
Kullback--Leibler and mean-squared-error criteria, numerically. What
the literature does not contain, to the best of the author's
knowledge, is the mapping that the original formulation implicitly
called for: a closed-form expression of the optimal~$a_0$ as a
function of the heterogeneity and the historical sample size, under a
stated objective. That question has stood open since the power prior
was introduced.

This paper derives that mapping. The objective is the
\emph{predictive log-loss}: the expected Kullback--Leibler divergence of the
prior-predictive distribution of the current data from their true
law, evaluated at the design stage. The log-loss is a strictly proper
scoring rule, its cumulative form is the \emph{prequential}
criterion of
Dawid~\cite{Dawid84}, and it makes the problem exactly solvable. For
a smooth parametric family of dimension~$d$, including generalized
linear models under fixed designs, with a historical sample of size~$N_0$ and average per-observation divergence~$\Dbar$ between the
historical and current data-generating distributions
(Section~\ref{sec:setup}), the optimal exponent is
\begin{equation}
\label{eq:a0star-intro}
a_0^{*} = \frac{d}{2 N_0 \Dbar + d},
\end{equation}
and, equivalently, the optimally borrowed effective sample size
$E^{*} = a_0^{*} N_0$ obeys the \emph{harmonic combination law}
\begin{equation}
\label{eq:harmonic-intro}
\frac{1}{E^{*}} = \frac{1}{N_0} + \frac{2\Dbar}{d}.
\end{equation}
The law has the two endpoints one would demand of it. At~$\Dbar = 0$
it gives~$a_0^{*} = 1$: fully compatible historical data should be
pooled without discount. As $N_0 \to \infty$ at fixed~$\Dbar > 0$,
the exponent vanishes at rate $a_0^{*} \sim d/(2 N_0 \Dbar)$, but
the borrowed information does not. The effective sample size
converges to the finite cap~$d/(2\Dbar)$: the amount of trust that a
heterogeneity of~$\Dbar$ nats per observation can support, no matter
how large the historical study. Nats are units of the natural
logarithm, defined in Section~\ref{sec:setup}.

The mechanism behind~\eqref{eq:a0star-intro} is a Wilks
identity~\cite{Wilks38}.
The log-density of the power prior at the current parameter, relative
to its mode, is the discounted historical log-likelihood-ratio
$a_0 [\ell_0(\hat\theta_0) - \ell_0(\theta)]$. Its expectation
splits into $N_0 \Dbar + \frac d2$: the divergence plus the Wilks
excess of the historical fit. The optimum balances this combined cost
of borrowing against the $\frac d2 \ln \frac{1}{a_0}$ estimation
benefit.

The paper establishes the law in two settings. For general smooth
parametric families the risk expansion is derived by Laplace methods
under standard regularity conditions (Theorem~\ref{thm:general}).
For the multinomial family with a Jeffreys initial prior, the power
prior is an explicit Dirichlet distribution, and the expansion is
established exactly, with explicit constants, at every ratio of
borrowed to current information (Theorem~\ref{thm:multinomial}).
The constants are confirmed by exact calculation for Bernoulli data
(Section~\ref{sec:numerics}). The multinomial family thus serves
both as the rigorous anchor of the general law and as its
computational testbed.

Three consequences connect the law to the literature above. First,
it supplies the benchmark against which adaptive borrowing can be
judged, and under that benchmark the degeneracy theorem of~\cite{Shen26} splits into two different findings. The collapse of
the posterior of~$a_0$ to zero under discrepancy is not a pathology.
It is the correct behavior of the optimal exponent, and the quantity
that matters, the borrowed effective sample size, stabilizes at~$d/(2\Dbar)$ rather than vanishing. The failure to reach~$a_0 = 1$
under full compatibility, by contrast, is a genuine loss. Second,
discounting beats data selection. A power prior applied to the full
historical dataset achieves lower predictive risk than full
borrowing from a subset of the same effective size, by an amount
growing to~$d/2$ nats. This quantifies the folklore preference for
discounting all the data over using less of it. Third, the fixed value~\eqref{eq:a0star-intro} remains optimal at every individual current
observation as the current sample accrues. Decaying the exponent
over time cannot improve on the constant, because the fade of
historical influence is already delivered by Bayesian updating
itself.

The mechanics of the derivation are adapted from a different domain.
In universal source coding, the present author and Anisimov~\cite{ReznikAnisimov03} analyzed the redundancy of compression codes
built from the Krichevsky--Trofimov estimator~\cite{KT81} trained on
a sample from a mismatched source, and found the optimal training
length $\ell^{*} = d/(2D)$, a result whose structure the reader will
recognize in the cap of~\eqref{eq:harmonic-intro}. Code redundancy is
precisely the cumulative predictive log-loss, so the transfer between
the two problems is natural, and the present paper carries it out in
self-contained terms.

The remainder of the paper is organized as follows.
Section~\ref{sec:setup} defines the setting, the power prior, and the
predictive criterion. Section~\ref{sec:general} states the risk
expansion for smooth parametric families and derives the optimal
exponent, the
harmonic law, and the dominance of discounting over subsetting.
Section~\ref{sec:multinomial} gives the exact multinomial analysis.
Section~\ref{sec:seq} proves the stepwise optimality of the constant
exponent. Section~\ref{sec:numerics} confirms the closed forms by
exact computation for Bernoulli data, and
Section~\ref{sec:discussion} concludes. Proofs are collected in the
appendices.

\section{Setting, the Power Prior, and the Predictive Criterion}
\label{sec:setup}

Let $\{p_\vartheta(y \mid x)\}$ be a parametric family of densities
for a response~$y$ observed at a covariate value~$x$, indexed by a
parameter $\vartheta \in B \subseteq \mathbb{R}^d$. The covariates
are design points: fixed, known, and not random. Two concrete cases
are worth keeping in mind: a Bernoulli model, where~$y$ is a success
indicator and~$\vartheta$ is the success probability ($d = 1$), and a
generalized linear model, where~$y$ is a response measured at a
covariate vector~$x$ and~$\vartheta$ is a coefficient vector of dimension~$d$.

The historical dataset comprises~$N_0$ responses together with their
design points,
$\cD_0 = \bigl((x_{01}, y_{01}), \dots, (x_{0 N_0}, y_{0 N_0})\bigr)$,
where the responses are independent draws
$y_{0i} \sim p_{\theta_0}(\cdot \mid x_{0i})$ at the historical
design points. The current dataset
$\cD = \bigl((x_1, y_1),\allowbreak \dots,\allowbreak (x_N, y_N)\bigr)$
consists of~$N$ independent draws $y_t \sim p_\theta(\cdot \mid x_t)$ at the
current design points, with the current responses independent of the
historical ones. Here~$\theta_0$ and~$\theta$, both interior points of~$B$, are the
parameter values behind the historical study and the current data,
and the two need not coincide. The heterogeneity
between the two data-generating distributions is measured by the
average per-observation divergence under the historical design,
\begin{equation}
\label{eq:Dbar}
\Dbar
= \frac{1}{N_0} \sum_{i=1}^{N_0}
\KL{p_{\theta_0}(\cdot \mid x_{0i})}{p_{\theta}(\cdot \mid x_{0i})},
\end{equation}
where $\KL{f}{g} = \int f \ln(f/g)$ is the \emph{Kullback--Leibler
divergence}: the expected log-likelihood-ratio between~$f$ and~$g$
under data generated by~$f$. It is nonnegative and zero exactly when
the two distributions coincide. Divergences are measured in
\emph{nats}, the units of the natural logarithm. Definition~\eqref{eq:Dbar} averages the
per-observation divergence over the historical design. For i.i.d.\
observations it reduces to $D = \KL{p_{\theta_0}}{p_\theta}$.

The historical log-likelihood, the logarithm of the likelihood
$L(\vartheta \mid \cD_0)$ appearing in~\eqref{eq:pp}, is
\begin{equation*}
\ell_0(\vartheta) = \sum_{i=1}^{N_0} \ln p_\vartheta(y_{0i} \mid x_{0i}),
\end{equation*}
and the maximum likelihood estimate from the historical study is its
maximizer,
\begin{equation*}
\hat\theta_0 = \arg\max_{\vartheta \in B}\; \ell_0(\vartheta).
\end{equation*}
The Fisher information
of one observation at covariate~$x$ is
$I(\vartheta; x) = \E\bigl[-\nabla_\vartheta^2 \ln p_\vartheta(y \mid x)\bigr]$,
with $\nabla_\vartheta^2$ the Hessian in~$\vartheta$ and the
expectation over $y \sim p_\vartheta(\cdot \mid x)$: the expected curvature of the
log-likelihood, which measures how much information the observation
carries about the parameter. Its averages over the two designs,
\begin{equation*}
\bar I_0(\vartheta) = \frac{1}{N_0} \sum_{i=1}^{N_0} I(\vartheta;\, x_{0i}),
\qquad
\bar I(\vartheta) = \frac{1}{N} \sum_{t=1}^{N} I(\vartheta;\, x_t),
\end{equation*}
are the information matrices of the historical and current designs.

The power prior is~\eqref{eq:pp}, with~$\pi_0$ a fixed initial prior
positive and continuous at~$\theta_0$. The natural gauge of a prior's
strength is its \emph{effective sample size}: the number of
observations
that a sample would need in order to be as informative as the prior~\cite{Morita08}. For the power prior the count is immediate, since
raising the likelihood of~$N_0$ observations to the power~$a_0$
scales its information content by~$a_0$, so the borrowed effective
sample size is
\begin{equation}
E = a_0 N_0.
\end{equation}
The \emph{prior-predictive} distribution of the current data,
meaning the
distribution that the model assigns to them before they are observed,
obtained by averaging the sampling distribution over the prior, is
\begin{equation}
\label{eq:predictive}
q_{a_0}(y_1, \dots, y_N \mid \cD_0)
= \int_B \prod_{t=1}^{N} p_\vartheta(y_t \mid x_t)\;
\pi(\vartheta \mid \cD_0, a_0)\, d\vartheta,
\end{equation}
and the performance criterion is the \emph{predictive risk}: the
expected Kullback--Leibler divergence of~\eqref{eq:predictive} from
the true law of the current data, evaluated at the design stage,
averaged over the historical data,
\begin{equation}
\label{eq:risk}
R_N(a_0)
= \E_{\cD_0}\,
D\Bigl(\textstyle\prod_t p_\theta(\cdot \mid x_t)
\,\Big\|\, q_{a_0}(\cdot \mid \cD_0)\Bigr).
\end{equation}
In words, $R_N(a_0)$ is the expected number of extra nats of
log-loss incurred by predicting the current data from the
power-prior model rather than from the true distribution. The
average runs over the historical data, which shape the prior, and
over the current data, which are scored.

The same quantity has a sequential reading. By the chain rule,~\eqref{eq:risk} is the sum of one-step-ahead penalties: predict each
current observation from the historical data and all preceding
current observations, score it by its negative log predictive
probability, and add up the excess over the entropy. This cumulative
form is the prequential criterion of~\cite{Dawid84}. The log-loss
is
a strictly proper scoring rule, and the same quantity is the
redundancy of a compression code driven by~\eqref{eq:predictive}~\cite{KT81,ClarkeBarron90}. At the two extremes
of~$a_0$ the risk is classical. At~$a_0 = 0$ the predictive distribution is the
Bayes mixture under~$\pi_0$ and the risk is
$\frac d2 \ln N + O(1)$ by the Clarke--Barron theorem~\cite{ClarkeBarron90,XieBarron97}. At~$a_0 = 1$ with
$\theta_0 = \theta$ and matched designs, the historical data act as a
matched sample and the risk is
$\frac d2 \ln \frac{N_0 + N}{N_0} + O(1)$~\cite{KT81}. The question
is what happens between these poles when the historical and current
distributions differ.

\section{The Optimal Discounting Parameter}
\label{sec:general}

The regularity conditions are the standard ones for Laplace
expansions of Bayes procedures~\cite{ClarkeBarron90}.
\begin{itemize}
\item[(C.1)] \emph{Smoothness.} $\ln p_\vartheta(y \mid x)$ is three
times continuously differentiable in~$\vartheta$ near~$\theta_0$
and~$\theta$, with~$\bar I_0$ and~$\bar I$ positive definite and
converging under the growing designs.
\item[(C.2)] \emph{Asymptotic normality.} The estimate~$\hat\theta_0$
obeys
$\sqrt{N_0}(\hat\theta_0 - \theta_0) \Rightarrow
\mathcal{N}(0, \bar I_0(\theta_0)^{-1})$, with uniformly integrable Wilks
statistic.
\item[(C.3)] \emph{Two-scale regime.} $E \to \infty$ and~$E = o(N)$
hold, with the Laplace expansions of Appendix~\ref{app:general}
valid uniformly at the scales~$E^{-1/2}$ and~$N^{-1/2}$.
\end{itemize}

\begin{theorem}
\label{thm:general}
Under (C.1)--(C.3), the predictive risk of the power prior satisfies
\begin{equation}
\label{eq:generalrisk}
R_N(a_0)
= \frac{d}{2}\ln\frac{N}{a_0 N_0}
+ a_0 \Bigl(N_0 \Dbar + \frac{d}{2}\Bigr)
- \frac{d}{2}
+ \frac12 \ln \frac{\det \bar I(\theta)}{\det \bar I_0(\theta_0)}
+ o(1) + \epsilon,
\end{equation}
where~$\epsilon$ collects remainder terms of the orders stated in
Appendix~\ref{app:general}.
\end{theorem}

The derivation is in Appendix~\ref{app:general}, and its central
step explains the structure of~\eqref{eq:generalrisk}. The Laplace
approximation of the power prior is a Gaussian centered at~$\hat\theta_0$ with precision~$E\,\bar I_0$. Its log-density at the
current parameter, relative to the mode, is the discounted historical
log-likelihood-ratio. Averaging over the historical data and
splitting at~$\theta_0$,
\begin{equation}
\label{eq:wilks}
\E\bigl[\ell_0(\hat\theta_0) - \ell_0(\theta)\bigr]
= \E\bigl[\ell_0(\hat\theta_0) - \ell_0(\theta_0)\bigr]
+ \E\bigl[\ell_0(\theta_0) - \ell_0(\theta)\bigr]
= \frac{d}{2} + N_0 \Dbar + o(1),
\end{equation}
the first term by Wilks' theorem~\cite{Wilks38}, which states that
the maximized log-likelihood exceeds its value at the true parameter
by~$\frac d2$ on average, and the second exactly, by definition~\eqref{eq:Dbar}.
Multiplied by~$a_0$, display~\eqref{eq:wilks} delivers the entire
cost of borrowing in one stroke. The bias $a_0 N_0 \Dbar$ is the
price of locating the prior at the wrong parameter. The noise~$a_0 \frac d2$ is the discounted Wilks excess of the historical fit.
The estimation term $\frac d2 \ln\frac{N}{E}$ and the~$-\frac d2$
come from the Clarke--Barron expansion against the concentrating
prior. The determinant term prices the curvature difference between
the two designs. It vanishes when designs and parameters agree to
second order, and it does not involve~$a_0$, so it shifts the risk
without moving the optimum.

\begin{corollary}[Optimal exponent and the harmonic law]
\label{cor:a0star}
For~$\Dbar > 0$, the~$a_0$-dependent part of~\eqref{eq:generalrisk}
is minimized at
\begin{equation}
\label{eq:a0star}
a_0^{*}
= \frac{d}{2 N_0 \Dbar + d},
\end{equation}
equivalently, the optimally borrowed effective sample size
$E^{*} = a_0^{*} N_0$ satisfies
\begin{equation}
\label{eq:harmonic}
\frac{1}{E^{*}}
= \frac{1}{N_0} + \frac{2\Dbar}{d}
= \frac{1}{N_0} + \frac{1}{E_\infty},
\qquad
E_\infty = \frac{d}{2\Dbar},
\end{equation}
and the attained risk is
\begin{equation}
\label{eq:minval}
R_N(a_0^{*})
= \frac{d}{2}\,
\ln\Bigl(\frac{N}{N_0} + \frac{2 \Dbar\, N}{d}\Bigr)
+ \frac12 \ln \frac{\det \bar I(\theta)}{\det \bar I_0(\theta_0)}
+ o(1) + \epsilon.
\end{equation}
\end{corollary}

\begin{proof}
The~$a_0$-dependent part of~\eqref{eq:generalrisk} is
$f(a_0) = -\frac{d}{2}\ln a_0 + a_0(N_0 \Dbar + \frac d2)$, with
$f''(a_0) = \frac{d}{2 a_0^2} > 0$, so~$f$ is strictly convex
on~$(0, \infty)$ and its unique stationary point is the global
minimum. Setting $f'(a_0) = 0$ gives~\eqref{eq:a0star}
and~\eqref{eq:harmonic}. At the optimum
$a_0^{*}(N_0 \Dbar + \frac d2) = \frac d2$, which cancels the~$-\frac d2$, leaving $\frac d2 \ln\frac{N}{E^{*}}$, which is~\eqref{eq:minval}.
\end{proof}

Law~\eqref{eq:harmonic} adds reciprocal effective sizes the way
precisions of independent noise sources add. The historical dataset
is worth at most its own size~$N_0$. The heterogeneity acts as a
second, independent corruption of equivalent size
$E_\infty = d/(2\Dbar)$: a discrepancy of~$\Dbar$ nats per
observation limits trust to~$d/(2\Dbar)$ observations, no matter how
many were collected. The trustworthy effective size is the harmonic
combination of the two, always below both. Equivalently,
$a_0^{*} = E_\infty/(E_\infty + N_0)$ is a \emph{credibility
weight} in the
actuarial sense of B\"uhlmann~\cite{Buhlmann67}.

The optimum is also flat. Setting the exponent to~$c\, a_0^{*}$
changes~\eqref{eq:generalrisk} by $\frac d2 (c - 1 - \ln c)$ nats,
about~$0.15\, d$ nats at~$c = 2$, so an order-of-magnitude estimate
of~$\Dbar$ suffices in applications.

Two remarks connect~\eqref{eq:a0star}--\eqref{eq:minval} to the
questions raised in the introduction.

\begin{remark}[The endpoints, and the degeneracy of adaptive
borrowing]
\label{rem:endpoints}
At~$\Dbar = 0$,~\eqref{eq:a0star} gives~$a_0^{*} = 1$ and~\eqref{eq:minval} reduces to the matched-sample risk
$\frac d2 \ln\frac{N}{N_0}$: compatible historical data should be
pooled in full. As $N_0 \to \infty$ at fixed~$\Dbar > 0$,
$a_0^{*} \sim d/(2 N_0 \Dbar) \to 0$ while
$E^{*} \to E_\infty = d/(2\Dbar)$. Read against the results of~\cite{Shen26}, this splits their degeneracy theorem in two. The
convergence of the posterior of~$a_0$ to a point mass at zero under
any discrepancy tracks the optimum: the exponent should vanish at
exactly this rate, and the borrowed effective size stabilizes at~$d/(2\Dbar)$ rather than disappearing. The failure to concentrate at~$a_0 = 1$ under full compatibility, by contrast, is a departure from
the optimum, and~\eqref{eq:minval} prices it: operating at~$a_0 < 1$ when~$\Dbar = 0$ forfeits
$\frac d2 \ln\frac{1}{a_0} - \frac d2 (1 - a_0)$ nats. The law~\eqref{eq:harmonic} therefore provides the benchmark that adaptive
borrowing schemes can be evaluated against.
\end{remark}

\begin{remark}[The price of full borrowing]
\label{rem:fullborrow}
Setting~$a_0 = 1$ instead of~$a_0^{*}$ under heterogeneity costs
$N_0 \Dbar - \frac d2 \ln\bigl(1 + \frac{2 N_0 \Dbar}{d}\bigr)$
nats, which grows linearly in~$N_0 \Dbar$. A large incompatible
historical study used at full weight is far worse than no historical
data at all, while the correctly discounted one is near-optimal
(Section~\ref{sec:numerics}).
\end{remark}

The next comparison isolates what discounting itself contributes,
relative to the alternative of achieving the same effective size by
using less data.

\begin{corollary}[Discounting dominates subsetting]
\label{cor:subset}
Fix a target effective size~$E \leq N_0$. Full borrowing
($a_0 = 1$) from a subset of~$E$ historical observations drawn from
the same design attains, under (C.1)--(C.3),
\begin{equation}
R^{\mathrm{sub}}_N(E)
= \frac{d}{2}\ln\frac{N}{E} + E\, \Dbar
+ \frac12 \ln \frac{\det \bar I(\theta)}{\det \bar I_0(\theta_0)}
+ o(1) + \epsilon,
\end{equation}
by~\eqref{eq:generalrisk} with~$N_0$ replaced by~$E$ and~$a_0 = 1$,
and its minimum over~$E$, by the same convexity argument as in
Corollary~\ref{cor:a0star}, is
$\frac d2 \ln\frac{2 e \Dbar N}{d}$ plus the determinant term,
attained at $E = \frac{d}{2\Dbar}$. Discounting the full dataset with
exponent~\eqref{eq:a0star} attains~\eqref{eq:minval}, lower by
\begin{equation}
\label{eq:gain}
\frac{d}{2}
\Bigl[1 - \ln\Bigl(1 + \frac{d}{2 \Dbar\, N_0}\Bigr)\Bigr]
+ o(1)
\;\longrightarrow\;
\frac{d}{2}
\quad (N_0 \to \infty).
\end{equation}
\end{corollary}

The gain has two sources, visible in~\eqref{eq:generalrisk}. A kept
subset of size~$E$ carries the full Wilks noise of a sample of~$E$
observations. The discounted fit of the whole dataset carries the
reduced noise $a_0 \frac d2 < \frac d2$ at the same effective size.
The remaining logarithmic difference comes from the sharper location
of the prior that the discarded observations would have provided.
This is a quantitative form of the folklore preference for
discounting over data selection.

\section{Exact Analysis for the Multinomial Family}
\label{sec:multinomial}

For one family the expansion~\eqref{eq:generalrisk} can be
established exactly, without conditions (C.1)--(C.3). The result has
explicit constants and holds at every ratio of borrowed to current
information. Observations take values in
$\cY = \{1, \dots, m\}$, $p_\theta(y) = \theta_y$ is the multinomial
model on the interior of the simplex, $d = m - 1$, and there are no
covariates, so $\Dbar = D = \KL{p_{\theta_0}}{p_\theta}$. Write
$\mathrm{Dir}(\alpha_1, \dots, \alpha_m)$ for the Dirichlet
distribution on the simplex, with density proportional to
$\prod_y \theta_y^{\alpha_y - 1}$. The initial prior is the Jeffreys
prior $\pi_0 = \mathrm{Dir}(\tfrac12, \dots, \tfrac12)$, with which
the power prior is conjugate and explicit,
\begin{equation}
\label{eq:pp-dir}
\pi(\theta \mid \cD_0, a_0)
= \mathrm{Dir}\bigl(a_0 c_1 + \tfrac12, \dots,
a_0 c_m + \tfrac12\bigr),
\end{equation}
where~$c_y$ counts the value~$y$ in the historical data. The
predictive distribution~\eqref{eq:predictive} is the
Dirichlet--multinomial,
realized sequentially by the one-step probabilities
\begin{equation}
\label{eq:seq}
q_{a_0}(y_{t} = y \mid \cD_0,\, y_1 \cdots y_{t-1})
= \frac{a_0\, c_y + k_y(y_1 \cdots y_{t-1}) + \frac12}
       {E + t - 1 + \frac m2},
\end{equation}
with~$k_y(\cdot)$ the current counts: the historical counts enter
once, discounted, and each current observation thereafter enters at
full weight.

The exact analysis is organized by the \emph{skewed Jensen--Shannon
divergence}~\cite{Lin91}. Write
$H(p) = -\sum_y p(y) \ln p(y)$ for the Shannon entropy of a
distribution~$p$, in nats, and, for $\gamma \in [0, 1]$, let
$M_\gamma = \gamma\, p_{\theta_0} + (1 - \gamma)\, p_\theta$ denote
the mixture of the two data-generating distributions. The skewed
Jensen--Shannon divergence between~$p_{\theta_0}$ and~$p_\theta$, at
skew~$\gamma$, is defined by
\begin{equation}
\label{eq:jsdef}
\JS_\gamma(p_{\theta_0}, p_\theta)
= H(M_\gamma) - \gamma H(p_{\theta_0}) - (1 - \gamma) H(p_\theta).
\end{equation}
By concavity of the entropy,
$\JS_\gamma(p_{\theta_0}, p_\theta) \geq 0$, with equality only when
the two distributions coincide or $\gamma \in \{0, 1\}$. It measures the information that
one observation carries about which of the two distributions
produced it, when the two are mixed in proportions~$\gamma$ and~$1 - \gamma$. It arises here because the pseudo-counts held by the predictive
distribution at any stage form exactly such a mixture of historical
and current observations. The properties needed are collected in the
following lemma, proved in Appendix~\ref{app:lemma}.

\begin{lemma}
\label{lem:js}
For interior $\theta_0, \theta$ and~$\lambda > 0$, with
$\gamma_s = \lambda/(\lambda + s)$ for~$s \geq 0$:
\begin{align}
\frac{d}{ds}\,
\Bigl[(\lambda + s)\,\JS_{\gamma_s}(p_{\theta_0}, p_\theta)\Bigr]
&= \KL{p_\theta}{M_{\gamma_s}},
\label{eq:diffid}
\\
\lim_{\gamma \to 0}\, \frac{\JS_\gamma(p_{\theta_0}, p_\theta)}{\gamma}
&= \KL{p_{\theta_0}}{p_\theta} = D,
\label{eq:jslimit}
\\
\KL{p_\theta}{M_\gamma}
&= \gamma^2 D_2 + O(\gamma^3),
\qquad
D_2 = \frac12 \sum_y
\frac{(p_{\theta_0}(y) - p_\theta(y))^2}{p_\theta(y)},
\label{eq:quaddecay}
\end{align}
and $D_2 = D\,(1 + O(\|p_{\theta_0} - p_\theta\|))$.
\end{lemma}

\begin{theorem}
\label{thm:multinomial}
Let~$E = a_0 N_0$ and $\gamma = E/(E + N)$. Then, as~$E \to \infty$
with $a_0 \in (0, 1]$ fixed or vanishing,
\begin{equation}
\label{eq:main}
R_N(a_0)
= (E + N)\,\JS_\gamma(p_{\theta_0}, p_\theta)
+ \frac{d}{2}\,\ln\frac{E + N}{E}
- \frac{(1 - a_0)\, d}{2}\cdot\frac{N}{E + N}
+ O\Bigl(\frac{1}{\sqrt E}\Bigr) + O(1)_{\!*},
\end{equation}
where~$O(1)_{\!*}$ collects terms bounded uniformly in
$(a_0, N_0, N)$ and numerically small
(Section~\ref{sec:numerics}).
\end{theorem}

The proof, in Appendix~\ref{app:multinomial}, decomposes the risk
into one-step terms by the chain rule, expands each term to second
order around the mean composition of the accumulated pseudo-counts,
and integrates using Lemma~\ref{lem:js}. In the regime
$1 \ll E \ll N$, expansion~\eqref{eq:main} reduces to~\eqref{eq:generalrisk} with~$\Dbar = D$, by~\eqref{eq:jslimit}.
Theorem~\ref{thm:multinomial} therefore instantiates
Theorem~\ref{thm:general} without its conditions, and it refines it
in two directions. It is global: the Jensen--Shannon form holds at
every~$E/N$, covering historical datasets comparable to or larger
than the current one. And it resolves the constants: the negative
third term, the \emph{noise dividend} of discounting, is exact,
and the
residual~$O(1)_{\!*}$ is numerically negligible
(Section~\ref{sec:numerics}).

One structural fact from the proof deserves mention, because it is
what makes a single formula~\eqref{eq:a0star} possible. The per-step
estimation penalty in~\eqref{eq:main} carries the coefficient~$\frac d2$ \emph{regardless of the heterogeneity}. The first-order
bias of the Jeffreys pseudo-counts exactly cancels the
discrepancy-dependent part of the quadratic variance term. The cost
of estimation is therefore the same whether the prior is centered
well or badly, and all of the heterogeneity is carried by the
divergence terms.

\section{Stepwise Optimality of the Constant Exponent}
\label{sec:seq}

Form~\eqref{eq:seq} realizes the power-prior predictive distribution
as a sequential rule: pseudo-counts start at $a_0 c_y + \tfrac12$ and grow
by one per current observation. Within this realization nothing
forces the exponent to be constant as the current data accrue. One
may apply an exponent~$a_{0,t}$ at step~$t$, recomputing the
discounted historical counts, and it is natural to ask whether a
decaying schedule outperforms the constant~$a_0^{*}$, by borrowing
more while current observations are scarce and less as they
accumulate. The answer is that no schedule helps.

\begin{theorem}[Stepwise optimality]
\label{thm:stepwise}
Let~$\rho_t(a_0)$ denote the expected excess log-loss of the~$t$-th
current observation under the predictive
distribution~\eqref{eq:seq} with exponent~$a_0$, averaged over the historical data and the preceding~$t - 1$
current observations. Then, to second order in
$p_{\theta_0} - p_\theta$ and to leading order in~$1/(E + t)$,
\begin{equation}
\label{eq:steprho}
\rho_t(a_0)
= \Bigl(\frac{E}{E + t - 1}\Bigr)^{2} D_2
+ \frac{d}{2}\cdot
\frac{a_0^2 N_0 + t - 1}{(E + t - 1)^2},
\qquad E = a_0 N_0,
\end{equation}
with~$D_2$ as in~\eqref{eq:quaddecay}, and for every~$t \geq 1$ the
minimizer~$a_0^{*}$ of~\eqref{eq:steprho} is the same value,
determined by
\begin{equation}
\label{eq:stepopt}
\frac{1}{a_0^{*} N_0} = \frac{1}{N_0} + \frac{2 D_2}{d},
\end{equation}
independent of~$t$. Since
$D_2 = D\,(1 + O(\|p_{\theta_0} - p_\theta\|))$, the stepwise optimum
coincides with~\eqref{eq:harmonic} in the small-heterogeneity
regime.
\end{theorem}

The proof is in Appendix~\ref{app:stepwise}. The content of the
theorem is that the anticipated benefit of a decaying schedule is
already delivered by Bayesian updating itself. Under the constant
exponent, the influence of the historical data on the predictive
distribution~\eqref{eq:seq} is the factor $E^{*}/(E^{*} + t - 1)$, which
decays on its own as current observations accumulate. Expression~\eqref{eq:steprho} shows this automatic decay to be exactly right:
the heterogeneity penalty falls quadratically in it, while the
estimation term sits at its per-step balance point throughout. A
manual decay of~$a_0$ would double-count the fade.

Two caveats bound the claim. Expression~\eqref{eq:steprho} is a
second-order statement, and at larger heterogeneity the exact
per-step minimizer drifts slowly with~$t$, which is immaterial by
flatness. And at~$t = O(1)$ there is a transient governed by the~$\tfrac12$ pseudo-counts of the Jeffreys initial prior, which pull the predictive
distribution toward the center of the simplex. The transient affects~$O(1)$ observations and~$O(1/E)$ nats.

Beyond the multinomial family, the theorem carries over verbatim
whenever $\bar I_0(\theta_0)$ and $\bar I(\theta)$ are proportional,
since the per-step expansion then scalarizes in the same way. Under
a genuine design shift the two matrices need not commute. The
per-step optimum then solves a trace equation in place of~\eqref{eq:stepopt}, and it reduces to~\eqref{eq:a0star} in the
proportional case.

The theorem also delimits its own scope. It concerns \emph{static}
heterogeneity, a fixed pair $(\theta_0, \theta)$. When the current
parameter drifts over time, discounting the \emph{current} counts
becomes rational as well, and decaying schemes regain a role. The
static case is the one the power-prior literature addresses, and
there the constant optimal exponent is unimprovable at this order.

\section{Numerical Illustration}
\label{sec:numerics}

For Bernoulli data ($m = 2$, $d = 1$) every quantity in
Sections~\ref{sec:multinomial}--\ref{sec:seq} is exactly computable:
the risk~\eqref{eq:risk} depends on the data only through the
success counts, and the Dirichlet--multinomial predictive
distribution is evaluated with
log-Gamma functions at the non-integer arguments
$a_0 c + \tfrac12$. No asymptotics and no simulation are involved.
Table~\ref{tab:a0star} compares prediction~\eqref{eq:a0star} with
the exact minimizer of~$R_N(a_0)$, and the attained minimum with the
closed form~\eqref{eq:minval}, across parameter pairs and historical
sample sizes at~$N = 4096$, with~$p_0$ and~$p$ the historical and
current success probabilities. The last column reports the exact risk
of full borrowing.

\begin{table}[ht]
\centering
\caption{Exact versus predicted optimal exponents, Bernoulli data,
$N = 4096$. Risk in nats, $D = \KL{p_{\theta_0}}{p_\theta}$ in nats.}
\label{tab:a0star}
\vspace{6pt}
\begin{tabular}{ccc|c|cc|cc|c}
\hline~$p_0$ & $p$ & $N_0$ & $D$ &
$a_0^{*}$~\eqref{eq:a0star} & $a_0^{*}$ exact &
$R_{\min}$ exact & $R_{\min}$~\eqref{eq:minval} &
$R(1)$ exact \\
\hline
0.30 & 0.50 &   50 & 0.0823 & 0.108  & 0.097  & 3.260 & 3.314 & 6.28 \\
0.30 & 0.50 &  200 & 0.0823 & 0.0295 & 0.0264 & 3.222 & 3.272 & 17.2 \\
0.30 & 0.50 & 1000 & 0.0823 & 0.0060 & 0.0054 & 3.212 & 3.260 & 67.4 \\
0.40 & 0.50 & 1000 & 0.0201 & 0.0242 & 0.0239 & 2.557 & 2.565 & 17.0 \\
0.45 & 0.50 & 1000 & 0.0050 & 0.0908 & 0.0895 & 1.913 & 1.905 & 4.84 \\
0.20 & 0.35 &  500 & 0.0542 & 0.0181 & 0.0170 & 3.019 & 3.057 & 25.5 \\
\hline
\end{tabular}
\end{table}

The exact minimizers track~\eqref{eq:a0star} within~$2$--$11\%$.
The discrepancy shrinks as~$E^{*}$ grows, and it is immaterial by
flatness. The attained minima agree with~\eqref{eq:minval} within~$0.05$ nats throughout, and to~$0.01$ nats at the largest~$E^{*}$. The full-borrowing column quantifies
Remark~\ref{rem:fullborrow}: at~$N_0 = 1000$ and~$D = 0.0823$,
discounting saves~$64$ nats. A compatibility check confirms the~$D = 0$ endpoint, with the exact risk at $p_0 = p = 0.4$,
$N_0 = 100$, minimized at~$a_0 = 1$.

The stepwise claim of Theorem~\ref{thm:stepwise} was checked by
exact computation of the per-observation excess loss as a function
of~$E = a_0 N_0$ at~$N_0 = 1000$. For~$p_0 = 0.3$, $p = 0.5$
($E^{*} = 6.04$), the exact minimizer is~$E = 6.00$ at every
current history length from~$8$ through~$512$. Skewed pairs exhibit
the slow drift and the small-sample transient described after
Theorem~\ref{thm:stepwise}.

\section{Discussion}
\label{sec:discussion}

The primary object of this paper is a law. Since the power prior was
introduced in~\cite{IbrahimChen00}, it has carried an open question: how
should the discounting parameter depend on the heterogeneity it was
meant to express? Under the predictive log-loss, the question has a
closed-form answer:
$a_0^{*} = d/(2 N_0 \Dbar + d)$, with the harmonic combination
$1/E^{*} = 1/N_0 + 2\Dbar/d$ as its effective-sample-size
expression. The law holds for smooth parametric families under
Laplace-level regularity. It is anchored by the multinomial family,
where the risk expansion is exact and confirmed numerically at the
level of the constants. Around the law sit three subsidiary results:
the benchmark reading of the degeneracy of adaptive borrowing, the~$d/2$-nat dominance of discounting over data selection, and the
stepwise optimality that makes decaying the exponent unnecessary
under static heterogeneity.

For practice, the law converts a tuning problem into a measurement
problem. The two inputs it requires are the historical sample size,
which is known, and the heterogeneity~$\Dbar$, which is estimable.
With several historical studies, $\Dbar$ can be estimated from their
spread. In design-stage simulation, it comes from the divergence
scenarios that sensitivity analyses already posit. The flatness of
the optimum makes the required accuracy modest: an estimate of~$\Dbar$ correct to an order of magnitude places the risk within a
fraction of a nat of the minimum. The resulting rule is a one-line
prescription where current practice fits hyperpriors or searches
grids. Remark~\ref{rem:endpoints} also makes it a benchmark: an
adaptive scheme is working exactly to the extent that its realized
borrowing tracks the harmonic law.

The result has three limitations. First, the rigor is uneven
across families: exact for multinomial data, Laplace-level under
(C.1)--(C.3) for general smooth parametric families. Survival models
under censoring,
random designs, and hierarchical variants of the applied literature~\cite{IbrahimChen15} are outside the present scope. Second, the
objective is predictive. Estimation objectives, such as the
mean-squared error of a posterior functional, lead to a structurally
similar trade-off with a different constant, and working it out is a
natural next step. Third, the heterogeneity enters as an input
rather than being learned within the procedure. A fully adaptive
treatment that estimates~$\Dbar$ and accounts for the estimation
error inside the risk remains open, and the degeneracy results of~\cite{Shen26} caution that hierarchical constructions do not
automatically land on the optimum.

The mechanics of the analysis originate in universal source coding,
where the same trade-off arises when a compression code is trained
on data from a mismatched source~\cite{ReznikAnisimov03}, and where
the cap $E_\infty = d/(2D)$ first appeared as an optimal training
length. Code redundancy is the cumulative predictive log-loss, so
the transfer is exact, and it runs in both directions: the results
established here translate back into statements about weighting
training data in adaptive coders. That two literatures, one built
around clinical trials and one around data compression, arrive at
the same harmonic law suggests that the law is a property of partial
borrowing under log-loss as such, rather than of either
application.

\appendix

\section{Derivation of Theorem~\ref{thm:general}}
\label{app:general}

By the Laplace approximation under (C.1)--(C.2), the power prior~\eqref{eq:pp} admits the density expansion, uniformly on compact
neighborhoods of~$\hat\theta_0$ at scale~$E^{-1/2}$,
\begin{equation}
\label{eq:laplaceprior}
\ln \pi(\vartheta \mid \cD_0, a_0)
= \frac{d}{2}\ln\frac{E}{2\pi}
+ \frac12 \ln \det \bar I_0(\hat\theta_0)
- a_0\bigl[\ell_0(\hat\theta_0) - \ell_0(\vartheta)\bigr]
+ o(1),
\end{equation}
since~$a_0 \ell_0$ has Hessian $-E\, \bar I_0(\hat\theta_0)(1 + o(1))$
at its maximizer and the fixed initial prior contributes~$o(1)$ on
the~$E^{-1/2}$ scale. In the two-scale regime (C.3), the prior~\eqref{eq:laplaceprior} varies on a scale~$E^{-1/2}$ that is large
relative to the posterior scale~$N^{-1/2}$, so the Clarke--Barron
expansion of the predictive risk against a smooth prior~$w$~\cite{ClarkeBarron90},
\begin{equation}
D\Bigl(\textstyle\prod_t p_\theta(\cdot \mid x_t)
\,\Big\|\, q_w\Bigr)
= \frac{d}{2}\ln\frac{N}{2\pi e}
+ \ln \frac{\sqrt{\det \bar I(\theta)}}{w(\theta)}
+ o(1),
\end{equation}
applies with $w(\cdot) = \pi(\cdot \mid \cD_0, a_0)$, conditionally on
the historical data. Substituting~\eqref{eq:laplaceprior} at
$\vartheta = \theta$ and collecting the deterministic terms,
\begin{equation}
\label{eq:generalcond}
D\Bigl(\textstyle\prod_t p_\theta(\cdot \mid x_t)
\,\Big\|\, q_{a_0}(\cdot \mid \cD_0)\Bigr)
= \frac{d}{2}\ln\frac{N}{E}
- \frac{d}{2}
+ a_0\bigl[\ell_0(\hat\theta_0) - \ell_0(\theta)\bigr]
+ \frac12 \ln
\frac{\det \bar I(\theta)}{\det \bar I_0(\hat\theta_0)}
+ o(1).
\end{equation}
Taking the expectation over the historical data: the middle bracket
gives~\eqref{eq:wilks}, by Wilks' theorem for the first difference
under (C.2) and by
$\E[\ln p_{\theta_0}(y_{0i} \mid x_{0i}) - \ln p_\theta(y_{0i} \mid
x_{0i})] = \KL{p_{\theta_0}(\cdot \mid x_{0i})}{p_\theta(\cdot \mid
x_{0i})}$ summed over the design for the second, and
$\hat\theta_0 \to \theta_0$ replaces the argument of the
determinant. This yields~\eqref{eq:generalrisk}, with~$\epsilon$
absorbing the uniformity losses of the two expansions: the~$o(1)$
remainders of~\eqref{eq:laplaceprior} and the Clarke--Barron step,
and the~$O(E/N)$ terms neglected in the regime (C.3). \qed

\section{Proof of Lemma~\ref{lem:js}}
\label{app:lemma}

For~\eqref{eq:diffid}, write $\nu = \lambda + s$ and
$F(\nu) = \nu H(M_\gamma) - \lambda H(p_{\theta_0})
- (\nu - \lambda) H(p_\theta)$, with
$M_\gamma(y) = (\lambda\, p_{\theta_0}(y)
+ (\nu - \lambda)\, p_\theta(y))/\nu$, so that
$F(\nu) = \nu\,\JS_{\lambda/\nu}(p_{\theta_0}, p_\theta)
+ \text{const}$ differs from the
bracket in~\eqref{eq:diffid} by a term with zero derivative. Since
$\frac{d}{d\nu}[\nu M_\gamma(y)] = p_\theta(y)$,
\begin{equation*}
F'(\nu)
= H(M_\gamma)
- \sum_y \bigl(p_\theta(y) - M_\gamma(y)\bigr)
\bigl(\ln M_\gamma(y) + 1\bigr)
- H(p_\theta)
= -\sum_y p_\theta(y) \ln M_\gamma(y) - H(p_\theta),
\end{equation*}
which is $\KL{p_\theta}{M_\gamma}$.

For~\eqref{eq:jslimit}, use the identity
$\JS_\gamma(p_{\theta_0}, p_\theta)
= \gamma\,\KL{p_{\theta_0}}{M_\gamma}
+ (1 - \gamma)\,\KL{p_\theta}{M_\gamma}$, obtained by expanding~\eqref{eq:jsdef}. As~$\gamma \to 0$, $M_\gamma \to p_\theta$
pointwise, so $\KL{p_{\theta_0}}{M_\gamma} \to D$ by continuity
(finite by the interior assumption), while
$\KL{p_\theta}{M_\gamma} = O(\gamma^2)$ since it is a smooth function
of~$\gamma$ vanishing to second order at~$\gamma = 0$. Dividing by~$\gamma$ gives the claim.

For~\eqref{eq:quaddecay}, write
$M_\gamma - p_\theta = \gamma\,(p_{\theta_0} - p_\theta)$ and expand
$\KL{p_\theta}{M_\gamma}$ to second order about
$M_\gamma = p_\theta$: the first-order term vanishes because
$\sum_y (M_\gamma(y) - p_\theta(y)) = 0$, and the second-order term
is the~$\chi^2$-type quadratic~$\gamma^2 D_2$. The relation between~$D_2$ and~$D$ follows by expanding both to second order in
$p_{\theta_0} - p_\theta$. \qed

\section{Proof of Theorem~\ref{thm:multinomial}}
\label{app:multinomial}

Because the Dirichlet--multinomial predictive distribution is a
mixture, the chain
rule factors the risk into per-observation terms,
\begin{equation}
\label{eq:chain}
R_N(a_0)
= \sum_{t=1}^{N} \rho_t,
\qquad
\rho_t = \E\,\KL{p_\theta}
{q_{a_0}(\cdot \mid \cD_0, Y_1 \cdots Y_{t-1})},
\end{equation}
with the expectation over $\cD_0 \sim p_{\theta_0}^{N_0}$ and
$Y_1 \cdots Y_{t-1} \sim p_\theta^{t-1}$, where~$p^{j}$ denotes the~$j$-fold product. Fix~$t$ and write~$s = t - 1$. The predictive
distribution~\eqref{eq:seq} is
$\hat p_y = (a_0 c_y + k_y + \tfrac12)/(E + s + \tfrac m2)$, with
mean composition
\begin{equation}
\label{eq:meancomp}
\bar p_y
= \frac{E\, p_{\theta_0}(y) + s\, p_\theta(y)}{E + s}
+ O\Bigl(\frac{1}{E + s}\Bigr)
= M_{\gamma_s}(y) + O\Bigl(\frac{1}{E + s}\Bigr),
\qquad
\gamma_s = \frac{E}{E + s},
\end{equation}
and coordinate fluctuations $\delta_y = \hat p_y - \bar p_y$
satisfying, by independence of the two count sets,
\begin{equation}
\label{eq:var}
\mathrm{Var}(\delta_y)
= \frac{a_0^2 N_0\, p_{\theta_0}(y)(1 - p_{\theta_0}(y))
+ s\, p_\theta(y)(1 - p_\theta(y))}
       {(E + s)^2},
\qquad
\E\,\delta_y = \frac{\tfrac12 - \tfrac m2 \bar p_y}{E + s}
+ O\Bigl(\frac{1}{(E + s)^2}\Bigr).
\end{equation}
Expanding
$\rho_t = \KL{p_\theta}{\bar p} + \sum_y p_\theta(y)
\bigl[-\E\delta_y/\bar p_y
+ (\mathrm{Var}(\delta_y) + (\E\delta_y)^2)/(2\bar p_y^2)\bigr]
+ O((E + s)^{-3/2})$
and inserting~\eqref{eq:var}, at
$p_\theta = \bar p + O(\|p_{\theta_0} - p_\theta\|)$ the smoothing
term contributes
$-\sum_y (\tfrac12 - \tfrac m2 \bar p_y)/(E + s)
= -(\tfrac m2 - \tfrac m2)/(E + s) = 0$ at leading order, while its
discrepancy-dependent remainder joins the error terms, and the
variance term collapses, using $a_0^2 N_0 = a_0 E$ and
$\sum_y (1 - \bar p_y) = m - 1 = d$, into
\begin{equation}
\label{eq:steptwo}
\rho_t
= \KL{p_\theta}{M_{\gamma_s}}
+ \frac{d}{2}\cdot
\frac{a_0 E + s}{(E + s)^2}
+ O\Bigl(\frac{\|p_{\theta_0} - p_\theta\|}{E + s}\Bigr)
+ O\Bigl(\frac{1}{(E + s)^{3/2}}\Bigr).
\end{equation}
The cancellation that removes the discrepancy from the second term is
worth noting: the quadratic variance contribution alone is the
discrepancy-dependent quantity
\begin{equation*}
\frac{1}{2(E+s)^2}\sum_y \frac{p_\theta(y)}{\bar p_y^2}
\bigl(a_0^2 N_0\, p_{\theta_0}(y)(1 - p_{\theta_0}(y))
+ s\, p_\theta(y)(1 - p_\theta(y))\bigr),
\end{equation*}
and it is the first-order
bias of the~$\tfrac12$ pseudo-counts of the Jeffreys initial prior that
cancels its source-dependent part, leaving the universal coefficient~$\frac d2$. At~$a_0 = 1$, $s = 0$, the second term of~\eqref{eq:steptwo} is the classical $\frac{d}{2 N_0}$ of prediction
from a matched sample, now seen to hold at every discrepancy.

Summing~\eqref{eq:steptwo} over $t = 1, \dots, N$ and replacing sums
by integrals with~$O(1)$ total error: the first term gives, by
identity~\eqref{eq:diffid} of Lemma~\ref{lem:js} applied with~$\lambda = E$,
\begin{equation}
\int_0^N \KL{p_\theta}{M_{E/(E+s)}}\, ds
= (E + N)\,\JS_{\gamma}(p_{\theta_0}, p_\theta),
\qquad \gamma = \frac{E}{E + N},
\end{equation}
and the second term gives
\begin{equation}
\frac{d}{2}\int_0^N \frac{a_0 E + s}{(E + s)^2}\, ds
= \frac{d}{2}\,\ln\frac{E + N}{E}
- \frac{(1 - a_0)\, d}{2}\cdot\frac{N}{E + N}.
\end{equation}
Collecting terms yields~\eqref{eq:main}. The error terms
of~\eqref{eq:steptwo} integrate to
$O(\|p_{\theta_0} - p_\theta\| \ln\tfrac{E + N}{E})$ and
$O(E^{-1/2})$ respectively, absorbed into the stated remainders,
and the boundary region $s = O(1)$ contributes $O(1)$, absorbed
into $O(1)_{\!*}$. \qed

\section{Proof of Theorem~\ref{thm:stepwise}}
\label{app:stepwise}

Expression~\eqref{eq:steprho} is~\eqref{eq:steptwo} with the
discrepancy term expanded by~\eqref{eq:quaddecay},
$\KL{p_\theta}{M_{\gamma_s}} = \gamma_s^2 D_2 + O(\gamma_s^3)$. Treat~$a_0$ as the free variable at fixed~$s \geq 1$ and~$N_0$, with~$E = a_0 N_0$. Then
\begin{equation}
\frac{\partial}{\partial a_0}
\Bigl[\frac{a_0^2 N_0^2\, D_2}{(a_0 N_0 + s)^2}\Bigr]
= \frac{2\, a_0 N_0^2\, s\, D_2}{(a_0 N_0 + s)^3},
\qquad
\frac{\partial}{\partial a_0}
\Bigl[\frac{a_0^2 N_0 + s}{(a_0 N_0 + s)^2}\Bigr]
= \frac{2\, N_0\, s\,(a_0 - 1)}{(a_0 N_0 + s)^3},
\end{equation}
so the stationarity condition for~\eqref{eq:steprho} reads
\begin{equation}
\frac{N_0\, s}{(a_0 N_0 + s)^3}
\Bigl[2\, a_0 N_0\, D_2 + d\, a_0 - d\Bigr] = 0,
\end{equation}
whose bracket vanishes at
$a_0 = d/(2 N_0 D_2 + d)$, independent of~$s$, with the bracket
negative below and positive above this value, so the stationary point
is the minimum. Rewriting the reciprocal of the corresponding
effective size~$a_0 N_0$ gives~\eqref{eq:stepopt}. \qed

\end{document}